\documentclass[letterpaper, 10 pt, conference]{ieeeconf}  
\IEEEoverridecommandlockouts             
\usepackage{amssymb} %

\usepackage{bm}
\usepackage{enumerate}
\usepackage{commath}
\usepackage{graphicx}
\graphicspath{{Pictures/}}
\usepackage{amsmath}
\usepackage{mathtools}

\usepackage{lipsum}  

\makeatletter
\let\MYcaption\@makecaption
\makeatother

\usepackage[font=footnotesize]{subcaption}

\makeatletter
\let\@makecaption\MYcaption
\makeatother

\usepackage{breqn}
\usepackage{cite}
\usepackage[table]{xcolor}
\usepackage{booktabs}
\usepackage{empheq}
\usepackage{amsfonts}

\usepackage{amsthm}
\usepackage{hyperref}
\usepackage{xcolor}
\usepackage{mathrsfs}
\usepackage{accents}
\usepackage{thmtools}
\usepackage{thm-restate}
\usepackage{float}
\usepackage{xcolor}
\usepackage[normalem]{ulem}

\usepackage{algorithm}
\usepackage{algpseudocode}
\usepackage{dblfloatfix}

\usepackage{etoolbox}
\usepackage{float}  %

\usepackage{comment}

\theoremstyle{plain}
\newtheorem{corollary}{Corollary}
\newtheorem{lemma}{Lemma}

\newtheorem{proposition}{Proposition}

\newtheorem*{theorem*}{Theorem}
\newtheorem{assumption*}{Assumption}

\declaretheorem[name=Theorem]{thm}

\theoremstyle{definition}

\newtheorem{remark}{Remark}
\newtheorem{definition}{Definition}
\newtheorem{assumption}{Assumption}

\newtheorem*{problem*}{Problem}

\newcommand{\myvar}[1]{#1}

\newcommand{\myset}[1]{\mathcal{#1}}

\DeclareMathOperator*{\argmin}{argmin}

\title{ \LARGE \bf
   Secure Safety Filter Design for Sampled-data Nonlinear Systems \\ under Sensor Spoofing Attacks
}

\author{
Xiao Tan, Pio Ong, Paulo Tabuada,  and Aaron D. Ames%
\thanks{This work is supported by TII under project \#A6847.}
\thanks{Xiao Tan, Pio Ong, and Aaron D. Ames are with the the Department of Mechanical and Civil Engineering, California Institute of Technology, Pasadena, CA 91125, USA (Email: {\tt\small xiaotan, pioong, ames@caltech.edu}).}    
\thanks{Paulo Tabuada is with the Department of Electrical and Computer Engineering at University of California, Los Angeles, CA 90095, USA (Email: {\tt\small  tabuada@ucla.edu}).} 
}

\begin{document}

\maketitle
\thispagestyle{plain}
\pagestyle{plain}

\begin{abstract}

This paper presents a secure safety filter design for nonlinear systems under sensor spoofing attacks. Existing  approaches primarily focus on linear systems which limits their applications in real-world scenarios. In this work, we extend these results to nonlinear systems in a principled way. We introduce exact observability maps that abstract specific state estimation algorithms and extend them to a secure version capable of handling sensor attacks. Our generalization also applies to the relaxed observability case, with slightly relaxed guarantees. More importantly, we propose a \textit{secure safety filter} design in both exact and relaxed cases, which incorporates secure state estimation and a control barrier function-enabled safety filter. The proposed approach provides theoretical safety guarantees for nonlinear systems in the presence of sensor attacks. We numerically validate our analysis on a unicycle vehicle equipped with redundant yet partly compromised sensors.
  \end{abstract}

\section{Introduction}
Safety of cyber-physical systems (CPS) has gained significant attention in the control community in recent years. Much of the existing literature focuses on designing control strategies to enforce safety specifications \cite{Ames2017,yu2024continuous}. However, a fundamental assumption in many of these works is that the system state is accessible to the the controller. In reality, as CPS 
 closely interact with external environments, they are susceptible to cyber-physical attacks that can compromise sensor measurements. Ensuring CPS safety under sensor attacks poses a critical challenge.

 A wide range of attack models have been studied in the literature, including, to name a few,  denial-of-service \cite{de2015input}, replay attacks \cite{zhu2013performance}, man-in-the-middle \cite{smith2015covert}, and false data injection \cite{mo2010false}. In this work, we focus on sensor spoofing attacks, where the attacker injects arbitrary signals into sensor measurements, without restrictions on their temporal, statistical, or bounded properties. The only restriction is that it can only attack at most certain number of sensors.
 Such attacks have been observed in real-world scenarios, including GPS spoofing in a drone incident \cite{rq170incident},  speed encoder attack on an vehicle\cite{shoukry2013non}, and acoustic interference with inertial measurement units (IMUs) on a robot \cite{tu2018injected}. If not properly handled,  corrupted measurements lead to incorrect state estimation, which compromises downstream planning and control tasks and may eventually put system safety at risk.

 A key theoretical challenge in this setting  is the secure state reconstruction problem -- determining when and how the system state can be recovered despite corrupted sensor measurements. Early results in this regard reported in \cite{fawzi2014secure,shoukry2015event} identify sparse observability property as a fundamental condition for exact secure state reconstruction for discrete-time linear systems. An equivalent condition is independently discovered in \cite{chong2015observability} for continuous-time linear systems. Our recent extension   \cite{tan2024safety} generalizes these results to allow for a finite set of plausible states under weaker sparse observability condition. All of these results, however, are restricted to linear systems. The only work that addresses the secure state reconstruction problem for nonlinear systems is \cite{shoukry2015secure}, which nevertheless assumes the system to be differentially flat.

Beyond secure state estimation, another important theoretical question is about system safety under sensor attacks. Our recent works \cite{tan2024safety,tan2025secure} introduce a \textit{secure safety filter}, which integrates secure state reconstruction with a control barrier function-based safety filter \cite{wabersich2023data}. This approach was experimentally validated on a quadrotor \cite{tan2025secure}, where a reduced-order linear model was used to approximate horizontal motion. However, extending this methodology to general nonlinear systems may be challenging and conservative.

\begin{figure}[t]
    \centering
    \includegraphics[width=\linewidth]{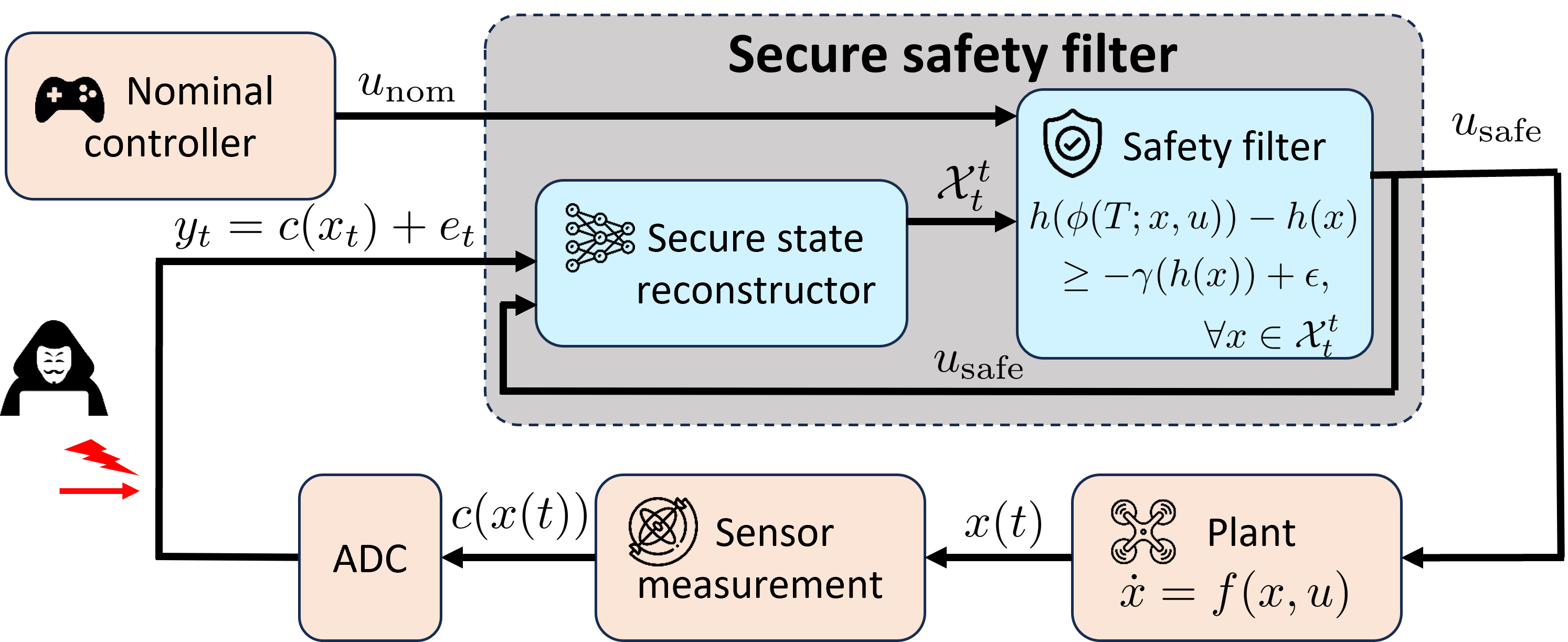}
    \caption{Secure safety filter diagram. The secure safety filter consists of a secure state reconstructor and a safety filter. The former takes in the input-output data, and calculates (an over-approximation of) the current plausible states. The latter then takes into account all plausible states to generate a safe control input with minimally invasive correction on the nominal input.}
    \label{fig:SSF}
\end{figure}

Several recent studies have explored safe control design for nonlinear systems subject to  sensor attacks with various assumptions. In  \cite{lin2023secondary}, a set of sensors is assumed to be attack-free, based upon which the safety filter is designed. The work in \cite{arnstrom2024data}  explicitly analyzes how stealthy sensors can  deactivate existing safety filters. It however requires an attacker capable of corrupting all sensors simultaneously. Another approach in \cite{zhang2022safe} divides sensors into groups, performs sensor anomaly detection per group, and applies a safety filter for states that pass the anomaly check. This approach requires more redundant sensors than necessary. Other works \cite{cosner2021measurement,lindemann2024learning}  propose robust safety filter designs for bounded measurement errors, but do not explicitly address adversarial attacks.  In contrast,  our work adopts a weaker sensor spoofing attack model with no explicit assumptions on the attack signals or the attacker's intention.

In this work, we propose a secure safety filter design  for general nonlinear systems under sensor spoofing attacks. The overall structure is shown in Figure \ref{fig:SSF}.  Our approach extends existing secure safety filter design \cite{tan2024safety,tan2025secure} to nonlinear systems. To achieve secure state reconstruction result for nonlinear systems, we introduce exact and relaxed observability maps that abstract specific state estimation algorithms, and generalize the corresponding the sparse observability notions.
The proposed approach provides provably safety guarantees for any nominal input and sensor spoofing attacks when the corresponding sparse observability condition and an online feasibility condition hold.

\textbf{Notation:} For $ w\in \mathbb{N}$, define $[w]:=\{1,2,\ldots,w\}$. The cardinality of a set $\myset{I}$ is denoted by $|\myset{I}|$. Given a $w\in \mathbb{N}$, a $k$-combination from $[w]$ is a subset of $[w]$ with cardinality $k$. Denote by $\mathbb{C}_{w}^k$ the set of all $k$-combinations from $[w]$. For a vector $y\in \mathbb{R}^w$ and an index set $\Gamma\subseteq[w]$, denote by $y^\Gamma$ the vector obtained by removing all entries not indexed in $\Gamma$. For a function $c:\mathbb{R}^n \to \mathbb{R}^w$ and an index set $\Gamma\subseteq[w]$, the map $c_{\Gamma}(\cdot)$ is similarly obtained by removing all the entries with indices not in $\Gamma$ from $c(\cdot)$.   A continuous function $\gamma: \mathbb{R} \to \mathbb{R}$ is an extended class $\mathcal{K}$ function if it is strictly increasing and $\gamma(0) = 0$. A ball $\mathbb{B}_\delta (z)$ of radius $\delta>0$ in $\mathbb{R}^n$ is defined as a set  $\mathbb{B}_\delta (z):= \{x\in \mathbb{R}^n: \| x - z \| \leq \delta\}$ with $z\in \mathbb{R}^{n}$.

\section{Problem Formulation}
\subsection{Nonlinear systems under sensor attacks}
Consider a continuous-time nonlinear system $\Sigma^C$
\begin{equation} \label{eq:continuous-time systems}
\Sigma^C: \left\{
\begin{aligned}
    & \dot{x}  = f(x, u) \\
    & y = c(x) + e
\end{aligned}    \right.
\end{equation}
where for $t\geq 0$, $x(t)\in \mathbb{R}^n, u(t)\in \mathbb{R}^m, y(t),  e(t)\in \mathbb{R}^p$ are the system state, the control input, the sensor measurement, and the attack signal to the system, respectively.  For digital implementation of controllers, we adopt the sampled-data control strategy where the input~$u$ is held constant between two consecutive sampling instants, that is,
\begin{equation} \label{eq:sample-and-hold input} 
    u(t) = u_k  \textup{ for } t\in [t_k, t_{k+1}).
\end{equation}
We assume a uniform sampling strategy with a sampling time $T=t_{k+1}-t_k$ for all $k\in\mathbb{N}$. We use the shorthand notations $x_k$, $y_k$, and $e_k$ for $x(t_k)$, $y(t_k)$, and $ e(t_k)$, respectively. Let $\phi(T;x,u)$ denote the end state of a trajectory of the system \eqref{eq:continuous-time systems} starting from $x$ and evolving for a time period $T$ with constant input $u$. The nonlinear system is assumed to be regular enough so that $\phi(T;x,u)$ also exists and is unique.

In this work, the nonlinear system in consideration is subject to sensor spoofing attacks. To simplify our analysis, the system is assumed to be equipped with $p$ independent sensors, whose outputs correspond to the $p$ entries of the measurement vector $y(t)$. When a sensor $i\in [p]$ is under a spoofing attack, the attacker can arbitrarily alter its measurement value $y^i(t)$ by injecting a nonzero $e^i(t)$.  Throughout this work, we assume the following attack model holds.

 \begin{assumption}[Attack model] \label{ass:s attacks}
     The attacker is omniscient, i.e., it has the complete knowledge of the system including the system states, the dynamics, and our defense strategy. However, it can only attack at most $s$ sensors, with the choices of sensors fixed. Mathematically, we require $$\left| \left\{i\in[p]:  \  \exists t\geq 0, \ e^i(t) \neq 0\right\} \right| \leq s.$$
 \end{assumption}

\subsection{System safety and secure safety filter}
For safety concerns, the nonlinear system state~$x$ is required to evolve within a safe set $\myset{C} = \{x\in \mathbb{R}^n: h(x) \geq 0\}$. One emerging tool for enforcing system safety is the control barrier function (CBF) framework \cite{Ames2017}. We assume that the safety constraint function $h$ is a zero-order CBF per the definition below.

\begin{definition}[Zero-order CBF \cite{tan2024zero}] \label{def:zocbf}
    A continuous  function $h:\mathbb{R}^n \to \mathbb{R}$ is a \emph{zero-order control barrier function}\footnote{Reference \cite{tan2024zero} defines zero-order CBFs for  state- and input-dependent safety constraints. Definition~\ref{def:zocbf} is a special case with state-only safety constraints. Although the results in this paper easily extend to input constraints, we focus on state constraints only for simplicity of presentation.}, if for all $x\in \mathbb{R}^n$, there exists an input $u\in \mathbb{R}^m$ such that:  
    \begin{equation}\label{eq:zocbf}
       h(\phi(T;x,u)) - h(x) \geq - \gamma(h(x)) + \epsilon,
    \end{equation}
where $\epsilon$ is a positive constant relating to the sampling time $T$ such that
$h(\phi(T;x,u)) \geq \epsilon $ implies $ h(\phi(t;x,u))\geq 0 $ for $t\in [0, T)$, and $\gamma(\cdot) $ is an extended class $\mathcal{K}$ function satisfying $| \gamma(s)|\leq |s|$ for $s\in \mathbb{R}$. 
\end{definition}
With the zero-order CBF, a safety filter for the sampled-data system is given by
\begin{equation} \label{eq:safety_filter}
    \begin{aligned}
       \myvar{u}_{\textup{safe}}(t) & = \argmin_u \| u - \myvar{u}_{\textup{nom}}(t_k) \| \\
       \textup{s.t. } & \eqref{eq:zocbf} \text{ holds with } x = x_k
    \end{aligned}, \textup{ for $t\in [t_k, t_{k+1})$.}
\end{equation}
This produces a safeguarding sample-and-hold control input $u_{\textup{safe}}(t)$ with a minimally invasive correction on a given nominal input $u_{\textup{nom}}(t)$ for $t\geq 0$.

Our work builds upon this zero-order CBF-based safety filter, which usually requires the system state at time $t_k$ available for feedback. Here, we design a safety filter for the sampled-data nonlinear system \eqref{eq:continuous-time systems}-\eqref{eq:sample-and-hold input} with corrupted measurements. 
For simplicity and with a slight abuse of notation, at time $t\geq 0, t\in [t_k, t_{k+1})$, we denote the input-output data from the last $(l+1)$ sampling instants  as $Y_{k-l:k} := \left(y_{k-l}, y_{k-l+1},  \ldots, y_{k}\right)$, $U_{k-l:k-1}:= \left(u_{k-l}, u_{k-l+1},  \ldots, u_{k-1}\right) $. In what follows, our proposed secure safety filter takes effects after collecting data from at least $(l+1)$ steps. We say the system~\eqref{eq:continuous-time systems} is \textit{safe} on set $\myset{C}$ with a control strategy determining $u$ if, given that the system state $x((l+1)T)$ is in $\myset{C}$, the system state $x(t)$ stays in the safety set $\myset{C}$ for all time $t\geq (l+1)T$.

In what follows, we will detail the design and the theoretical guarantees of our proposed secure state filter under two different observability assumptions.

\section{Exact Observability}
In this section, we consider a discrete-time system 
\begin{equation} \label{eq:sampled-data_system no disturbance}
\Sigma: \left\{
\begin{aligned}
    x_{k+1} & = F(x_k, u_k),  \\
     y_{k+1}  & = c(x_{k+1}) + e_{k+1}.
\end{aligned}    \right.
\end{equation}
Here  $F(x_k, u_k):=\phi(T;x_k,u_k)$ denotes the state transition map during one sampling period, which is assumed to be known. We introduce the following observability notion. 
\begin{definition}[Differential observability] \label{def:diff_observability}
    A sampled-data system in \eqref{eq:sampled-data_system no disturbance} is \emph{differentially observable} of order $l$ if, when no sensor attack is present, i.e., $ e_{k} = 0$ for all $k$, there exists a map $\mathcal{L}: \mathbb{R}^{ml}\times \mathbb{R}^{p(l+1)}\to \mathbb{R}^n$ such that
    \begin{equation}
        x_{k-l} =  \mathcal{L}\left(U_{k-l:k-1}, Y_{k-l:k}\right),
    \end{equation}
    where $x_{k-l} $ is the the solution of \eqref{eq:sampled-data_system no disturbance} at time step $k-l$.
\end{definition}

This notion is an analog to the differential observability notion for continuous-time systems \cite[Chapter 5]{bernard2019observer}, which states that the system state can be instantaneously determined using high-order time derivatives of the input and the measurement signals. The slightly stronger notion of differential flatness was utilized in a previous work \cite{shoukry2015secure} to solve the secure state reconstruction problem for discrete-time systems. Here the observability map $\mathcal{L}$ is an abstraction of any particular state estimation method that may be applicable to specific systems, for example, solving a system of equations for linear, observable discrete-time systems \cite{tan2024safety}. 

Typically, the observability map $\mathcal{L}$ does not take the possibility of attacks into account. When attacks are present, i.e., $ \exists j\in \mathbb{N}, e_j\neq 0$, the map $\mathcal{L}$ may produce a wrong state estimate because $Y_{k-l:k}$ is corrupted. A more sophisticated state estimation should consider initial states that are plausible given input-output data under sensor attacks.

\begin{definition}[Plausible initial states] \label{def:plausible state no disturbance}
     Given input-output data $(U_{k-l:k-1}, Y_{k-l:k})$ of the system \eqref{eq:sampled-data_system no disturbance}, we call  $z_{k-l}$ \emph{a plausible initial state} if there exists $\{e_j\}_{j=k-l}^k$ satisfying Assumption \ref{ass:s attacks} and the following system of equations:
     \begin{equation} \label{eq:plausible condition no disturbance}
    \begin{aligned}
        & z_{j+1} = F(z_j, u_j) ,~   j= k-l, \ldots, k-1,\\
        & y_{j} = c(z_{j}) + e_j,~   j= k-l, \ldots, k.
    \end{aligned}
\end{equation}  
  We denote the set of all plausible initial states by $\myset{X}_k^{k-l}$.
 \end{definition}

The initial state $x_{k-l}$ is one of the plausible states because it must satisfy the system dynamics for some $\{e_j\}_{j=k-l}^k$. Since this attack signal is unknown, we cannot distinguish between the correct state and other plausible ones. In order to identify the plausible states, the system design must incorporate sensor redundancies, allowing it to generate state estimates without some sensors.

 To formalize this, 
we first define systems with partial observations.  Let $\Gamma \subseteq [q]$ be an index set of the sensors. The system $\Sigma_{\Gamma}$ with measurements from $\Gamma$ is given by
\begin{equation} \label{eq:disturbance-free gamma-system}
\Sigma_{\Gamma}: \left\{
\begin{aligned}
    x_{k+1} & = F(x_k, u_k),  \\
     y_{k+1}^{\Gamma}  & = c_{\Gamma}(x_{k+1}) + e_{k+1}^\Gamma,
\end{aligned}    \right.
\end{equation}
where $c_{\Gamma}(x_{k+1})$ and $e_{k+1}^\Gamma$ are obtained  by removing all the entries with indices not in $\Gamma$ from $c(x_{k+1})$ and $ e_{k+1}$, respectively. Let $Y_{k-l:k}^\Gamma =  \left(y_{k-l}^\Gamma , y_{k-l+1}^\Gamma ,  \ldots, y_{k}^\Gamma \right) $.

Suppose that  the system \eqref{eq:disturbance-free gamma-system} is  differentially observable. An observability map  $\mathcal{L}_\Gamma$  thus exists by definition. However, even if possible ``inconsistencies" exist in the input-output data $(U_{k-l:k-1}, Y_{k-l:k}^{\Gamma})$ due to sensor attacks,  the map $\mathcal{L}_\Gamma$ always provides a state estimate. To better characterize how reasonable the state estimate is, we introduce the following consistency condition.

\begin{definition}[Consistency condition]
    Suppose that the system \eqref{eq:disturbance-free gamma-system} is  differentially observable. The input-output data $(U_{k-l:k-1}, Y_{k-l:k}^\Gamma)$ of the system \eqref{eq:disturbance-free gamma-system} is \emph{consistent} if  variables $(z_{k-l}, \ldots,z_{k} )$ exist such that the following equations hold
    \begin{equation} \label{eq:consistency no disturbance}
    \begin{aligned}
        & z_{k-l} =  \mathcal{L}_{\Gamma}\left( U_{k-l:k-1},Y_{k-l:k}^\Gamma \right), \\
        & z_{j+1} = F(z_j, u_j) ,   \hspace{0.6cm} \text{for } j= k-l, \ldots, k-1,\\
        & y_{j}^{\Gamma} = c_{\Gamma}(z_{j}) ,  \hspace{1.4cm} \text{for } j= k-l, \ldots, k.
    \end{aligned}
\end{equation}    
\end{definition}
In plain words, the data is consistent if the observability map produces a past state that, when propagated through the system model~\eqref{eq:disturbance-free gamma-system} to the current time, agrees with the data. The consistency condition is useful for identifying state estimates that are not compatible with the dynamics and measurement model. Nevertheless, the omniscient attacker can manipulate the measurements to yield a consistent yet incorrect state estimate. To this end, our strategy will involve checking state estimates from different sets of sensors. Therefore, we introduce the notion of sparse observability that ensures we can produce state estimates for different subset of sensors.
\begin{definition}[$r$-Sparse observability] \label{def:sparse_observability no disturbance}
   The system \eqref{eq:sampled-data_system no disturbance} is  \emph{$r$-sparse (differentially) observable} if,  for any $\Gamma \in \mathbb{C}_{p}^{p-r}$, the system $\Sigma_{\Gamma}$ is differentially observable of some order $l$. 
\end{definition}

The sparse observability property for a given system can be checked offline. This property is closely related to how many sensor attacks the system can endure. We first present a lemma that will be useful in later proofs.

\begin{lemma}\label{lem:consistency over two Gamma}
     Let $\Gamma_1,\Gamma_2$ be two index sets and $\Gamma_1 \subseteq \Gamma_2 \subseteq [p]$. We have the following results:
     \begin{enumerate}
         \item If the system $\Sigma_{\Gamma_1}$ is differentially observable, then the system $\Sigma_{\Gamma_2}$ is also differentially observable;
         \item Suppose that the systems $\Sigma_{\Gamma_1}$ and $\Sigma_{\Gamma_2}$ are both differentially observable with observability maps $ \mathcal{L}_{\Gamma_1}, \mathcal{L}_{\Gamma_2}$, respectively.  If the input-output data $(U_{k-l:k-1}, Y_{k-l:k}^{\Gamma_2} )$ is consistent for $\Sigma_{\Gamma_1}$, then the input-output data $(U_{k-l:k-1}, Y_{k-l:k}^{\Gamma_1} )$ is also consistent for $\Sigma_{\Gamma_1}$, and the following relation holds:
     \begin{equation}
        \hspace{-5mm} \mathcal{L}_{\Gamma_1}(U_{k-l:k-1}, Y_{k-l:k}^{\Gamma_1} ) = \mathcal{L}_{\Gamma_2}(U_{k-l:k-1}, Y_{k-l:k}^{\Gamma_2} ).
     \end{equation}
     \end{enumerate}
\end{lemma}
\begin{proof}
    We first prove Statement 1). One can verify by definition that the map $$\mathcal{L}_{\Gamma_2}(U_{k-l:k-1}, Y_{k-l:k}^{\Gamma_2}) = \mathcal{L}_{\Gamma_1}(U_{k-l:k-1}, Y_{k-l:k}^{\Gamma_1} )  $$ is a suitable observability map for the system $\Sigma_{\Gamma_2}$.

     Next, we prove Statement 2) via contradiction. Let  $z_{k-l}^1 =  \mathcal{L}_{\Gamma_1}(U_{k-l:k-1}, Y_{k-l:k}^{\Gamma_1} ), z_{k-l}^2 =  \mathcal{L}_{\Gamma_2}(U_{k-l:k-1}, Y_{k-l:k}^{\Gamma_2} )$. Suppose that $z_{k-l}^1 \neq z_{k-l}^2 $. Since  the input-output data $(U_{k-l:k-1}, Y_{k-l:k}^{\Gamma_2} )$ is consistent, $z_{k-l}^2 $ and its corresponding state sequence $(z_{k-l}^2, \ldots,z_{k}^2 )$ solves \eqref{eq:consistency no disturbance}. As the measurement data $Y_{k-l:k}^{\Gamma_1}$ is part of $Y_{k-l:k}^{\Gamma_2}$, the state sequence $(z_{k-l}^2, \ldots,z_{k}^2 )$ also solves the consistency condition \eqref{eq:consistency no disturbance} for  $(U_{k-l:k-1}, Y_{k-l:k}^{\Gamma_1} )$. Thus the input-output data $(U_{k-l:k-1}, Y_{k-l:k}^{\Gamma_1} )$ is also consistent. Consider a virtual system with dynamics 
     $$
     z_{k+1} = F(z_k,u_k), ~y^{\Gamma_1}_k = c_{\Gamma_1}(z_k).
     $$ Then two distinct states $z_{k-l}^1,z_{k-l}^2 $ are both possible initial states for input-output data $(U_{k-l:k-1}, Y_{k-l:k}^{\Gamma_1} )$, which contradicts with the differential observability assumption in the premise. Thus, we conclude that $z_{k-l}^1 = z_{k-l}^2 $.
\end{proof}

Now we are ready to present our first main result that establishes the equivalence between the collection of initial state estimates that are consistent with data and the set of plausible initial states.
\begin{proposition} \label{prop:s_sparse no disturbance}
    If system \eqref{eq:sampled-data_system no disturbance} is $s$-sparse observable, then the set of plausible states satisfies
    \begin{equation} \label{eq:X_t_t-l no disturbance}
    \begin{aligned}
     \hspace{-2mm}   \myset{X}_{k}^{k-l} = \bigcup_{\Gamma \in \mathbb{C}_{p}^{p-s}}& \left\{ \mathcal{L}_{\Gamma} ( U_{k-l:k-1}, Y_{k-l:k}^\Gamma ) \right.  \\ 
        &  \left. \text{ if } (U_{k-l:k-1}, Y_{k-l:k}^\Gamma) \text{ is consistent}\right\}.
    \end{aligned}
    \end{equation}
\end{proposition}
\begin{proof}
    We prove this result in two steps. For notational simplicity, let $\myset{S}_{RHS}$ be the set on the right-hand side of \eqref{eq:X_t_t-l no disturbance}. We will show that $\myset{X}_k^{k-l} \subseteq \myset{S}_{RHS}$ and $\myset{S}_{RHS} \subseteq \myset{X}_k^{k-l}$.
    
   $\myset{X}_k^{k-l} \subseteq \myset{S}_{RHS}$:  Consider any $z_{k-l} \in \mathcal{X}_k^{k-l}$. Based on Assumption \ref{ass:s attacks}, there exists a corresponding sequence of attacking signals $\{e_j\}_{j=k-l}^k$. Denote the attacked sensors as $\Gamma_A$. From Definition  \ref{def:plausible state no disturbance}, we know that $|\Gamma_A| \leq s$. Let $\Gamma $ be a subset of $[p]\setminus \Gamma_A$ with cardinality $p-s$. Recall that the system is assumed to be $s$-sparse observable, and thus, an observability map $ \mathcal{L}_{\Gamma}$ exists. Because the input-output data $ (U_{k-l:k-1}, Y_{k-l:k}^\Gamma)$ is not corrupted by the attack signals $e_k$, it  must be consistent. Thus we conclude $z_{k-l} =\mathcal{L}_{\Gamma}(U_{k-l:k-1}, Y_{k-l:k}^\Gamma)\in \myset{S}_{RHS}$. 

    $ \myset{S}_{RHS} \subseteq  \myset{X}_k^{k-l} $: Consider any $z_{k-l} \in \myset{S}_{RHS}$ with a corresponding $\Gamma \in \mathbb{C}_{p}^{p-s}$ and input-output data $(U_{k-l:k-1}, Y_{k-l:k}^\Gamma)$. Let $(z_{k-l}, \ldots,z_{k} )$ be the state sequence starting from $z_{k-l}$ with input sequence $U_{k-l:k-1}$.  Letting $e_j = y_j - c(z_{j})$ for $j = k-l, \ldots, k$, we know by consistency condition that $e_j^i = 0$ for $i\in \Gamma$. Thus, the sequence of attack signals $\{e_j\}_{j=k-l}^k$  satisfies Assumption \ref{ass:s attacks} and the condition \eqref{eq:plausible condition no disturbance}. By definition, we conclude that $z_{k-l} \in \myset{X}_k^{k-l} $.
\end{proof}

In general, $s$-sparse observability assumption is necessary for attacks to be detectable, i.e., that we can 
distinguish scenarios where all of sensors are attack-free or some are under attacks. See \cite[Theorem 16.1]{diggavi2020coding}  for further details. From Proposition \ref{prop:s_sparse no disturbance}, we know that the set $\myset{X}_{k}^{k-l}$ is finite under Assumption~\ref{ass:s attacks}. More importantly, Proposition \ref{prop:s_sparse no disturbance} offers an approach to compute the set of plausible states.

In the following we show that with higher sensor redundancy level, we are able to determine the correct state.

\begin{corollary} \label{cor:2s_sparse no disturbance}
    If system \eqref{eq:sampled-data_system no disturbance} is $2s$-sparse observable, then the state $x_{k-l} $ satisfies 
    $$x_{k-l} = \mathcal{L}_{\Gamma}\left( U_{k-l:k-1}, Y_{k-l:k}^\Gamma \right) $$ whenever $( U_{k-l:k-1}, Y_{k-l:k}^\Gamma)$ is consistent and $\Gamma \in \mathbb{C}_{p}^{p-s}$.
\end{corollary}
\begin{proof}
    From Statement 1 of Lemma~\ref{lem:consistency over two Gamma}, we deduce that $2s$-sparse observable systems are also $s$-sparse observable. Therefore, from
    Proposition \ref{prop:s_sparse no disturbance}, the state $x_{k-l}\in \mathcal X_k^{k-l}$ as defined in \eqref{eq:X_t_t-l no disturbance}.  
     We will show that, for any $\Gamma\in \mathbb{C}_{p}^{p-s}$ such that $(U_{k-l:k-1}, Y_{k-l:k}^\Gamma)$  is consistent,
     the observability map will produce the correct state as $\mathcal{L}_{\Gamma}\left( U_{k-l:k-1}, Y_{k-l:k}^\Gamma \right) = x_{k-l}$. Let $z_{k-l} = \mathcal{L}_{\Gamma}\left( U_{k-l:k-1}, Y_{k-l:k}^\Gamma \right)$ be the state produced by the observability map.
     Because the sensor index set $\Gamma$ has cardinality $p-s$, and there can be at most $s$ attacked sensors, there exists an attack-free subset of sensors $\Gamma^\prime \subset \Gamma$ with $|\Gamma^\prime| = p-2s$. From the $2s$-sparse observability assumption, there also exists a corresponding observability map $\mathcal{L}_{\Gamma^\prime}$.
     The observability definition requires that $\mathcal{L}_{\Gamma^\prime}( U_{k-l:k-1}, Y_{k-l:k}^{\Gamma^\prime} ) = x_{k-l}$ because there is no attack on sensors in ${\Gamma^\prime}$, i.e., $e_k^{\Gamma^\prime}\equiv 0$. At the same time, $( U_{k-l:k-1}, Y_{k-l:k}^{\Gamma^\prime} )$ is consistent and we conclude $\mathcal{L}_{\Gamma}( U_{k-l:k-1}, Y_{k-l:k}^{\Gamma} ) = \mathcal{L}_{\Gamma^\prime}( U_{k-l:k-1}, Y_{k-l:k}^{\Gamma^\prime} ) = x_{k-l}$ from Lemma~\ref{lem:consistency over two Gamma},  concluding that $x_{k-l}=z_{k-l}$ as desired.
\end{proof}

Once the set of initial plausible states $\myset{X}_{k}^{k-l}$ is obtained, we can propagate it forward along system dynamics to compute the current plausible states
\begin{equation}\label{eq:propagation no disturbance}
\begin{aligned}
    \myset{X}_{k}^{k} & = \{ z_k\in \mathbb{R}^n: \exists z_{k-l}\in \myset{X}_{k}^{k-l}  \\ 
    & \textup{ s.t. } z_{j+1}= F(z_{j},u_j) \textup{ for } j = k-l, \ldots, k-1  \}.
\end{aligned}
\end{equation}
A secure safety filter for time  $t\in [t_k, t_{k+1})$ is given by
\begin{equation} \label{eq:ssf no disturb}
    \begin{aligned}
        &  u_{\textup{safe}}(t) =    \argmin_{u} \ \|  u - u_{\text{nom}}(t_k)\|^2 \\
        \mathrm{s.t. } \ &   h(F(x,u)) - h(x) \geq - \gamma(h(x)) + \epsilon  \textup{ for } x\in \myset{X}_{k}^{k}.
    \end{aligned}
\end{equation}

\begin{thm}\label{thm:safe_control no disturbance}
   Suppose that Assumption~\ref{ass:s attacks} holds and $h$ is a zero-order CBF for system~\eqref{eq:continuous-time systems}.
   When system \eqref{eq:plausible condition no disturbance} is $s$-sparse observable,  system~\eqref{eq:continuous-time systems} is safe on set $\myset{C}$ with the secure safety filter~\eqref{eq:ssf no disturb} if the filter is always feasible. When system \eqref{eq:plausible condition no disturbance} is $2s$-sparse observable, system~\eqref{eq:continuous-time systems} is safe on set $\myset{C}$ with the secure safety filter~\eqref{eq:ssf no disturb}.
\end{thm}

\begin{proof}
    When the system is $s$-sparse observable, since the state $x_{k-l}\in \myset{X}_{k}^{k-l} $, we know that $x(t_k)\in \myset{X}_{k}^{k} $ based on~\eqref{eq:propagation no disturbance}. Under the feasibility assumption, we obtain $h(F(x(t_k),u))\geq (1- \gamma)( h(x(t_k))) + \epsilon  \geq \epsilon$. From the definition of an zero-order CBF, we thus deduce $h(x(t)) \geq 0$ for $t\in [t_k, t_{k+1})$. Applying this analysis recursively, we obtain safety guarantee for the closed-loop continuous-time system. When the system is $2s$-sparse observable, the feasibility property is guaranteed thanks to that $h$ is a zero-order CBF and that $ \myset{X}_{k}^{k}$ is a singleton set containing the state $x(t_k)$  from Corollary~\ref{cor:2s_sparse no disturbance}.
\end{proof}

When the system is $s$-sparse observable, it is difficult, in general, to verify the feasibility of the secure safety filter \eqref{eq:ssf no disturb} \textit{a priori}. Intrinsically, this is due to 
$s$-sparse observability being fairly weak when arbitrary $s$ sensors can be under attack. If the system has higher level of sensor redundancy or smaller number of sensor attacks,  we can establish stronger feasibility claims, see also the detailed discussions on the necessary and sufficient conditions for guaranteeing feasibility for linear systems in \cite{tan2024safety}.

\begin{remark}
   When the discrete-time system \eqref{eq:sampled-data_system no disturbance} is linear, then, with data length $l+1 \geq  n$, the differential observability notion reduces to classic observability notion, and the consistency condition becomes a linear equation satisfaction condition. All results in this section have a corresponding explicit form in the linear system case. Interested readers are referred to \cite{tan2024safety} for a concrete discussion. As illustrated here, the observability and the consistency notions are the key enabler for extending these results to nonlinear systems.  However, the exact differential observability definition can be challenging to satisfy for general nonlinear systems, and we relax it in the following section.
\end{remark}

\section{Relaxed Observability}
Now we consider a  sampled-data system formulation with process disturbance $w$
\begin{equation} \label{eq:sampled-data_system disturbance}
\Sigma^D: \left\{
\begin{aligned}
    x_{k+1} & = F(x_k, u_k) + w_k,  \\
     y_{k+1}  & = c(x_{k+1})  + e_{k+1},
\end{aligned}    \right.
\end{equation}
where $F(x_k, u_k)\approx \phi(T;x_k,u_k)$  approximates state transition map after one sampling period, which is assumed to be known. The process disturbance $w_k\in \mathbb{R}^n$ represents possible error due to the approximate state transition map for the sampling period $[t_k,t_{k+1})$. We assume that the disturbance is bounded:
\begin{equation} \label{eq:w_k bound}
    \| w_k \|\leq \bar{w}, \text{ for all } k.
\end{equation}

 \begin{remark} \label{rem:delta-bounded observability}
    In many cases, the function $F(\cdot,\cdot)$ is obtained using numerical integration methods (forward Euler, Runge-Kutta, etc). The error bound $\bar{w}$ is related to the sampling time $T$ and the specific numerical integration techniques. For example, for a sufficiently smooth vector field $f(\cdot,\cdot)$ with the classic fourth-order Runge-Kutta integration method approximating the flow, there exists a constant $K>0$ such that  $ \bar{w}\leq  K T^5$ when $T$ is small enough \cite{stuart1998dynamical}. This implies that the upper bound $\bar{w}$ can be enforced to be  arbitrarily small by choosing a small enough sampling time $T$.

 \end{remark}

We introduce a notion of differential observability for system~\eqref{eq:sampled-data_system disturbance} with unknown bounded disturbance.
\begin{definition}[$\delta$-Bounded differential observability] \label{def:delta_diff_observability}
    A sampled-data system in \eqref{eq:sampled-data_system disturbance} is \emph{$\delta$-bounded differential observable} of order $l$ if, when no sensor attack is present, i.e., $ e_{k} = 0$ for all $k$, there exist a set-valued map $\mathcal{L}^D: \mathbb{R}^{p(l+1)}\times \mathbb{R}^{ml}\to 2^{\mathbb{R}^n}$ and a  ball 
 $\mathbb{B}_\delta (\hat{x}_{k-l})$ such that
    \begin{equation}
        x_{k-l} \in  \mathcal{L}^D\left(U_{k-l:k-1}, Y_{k-l:k}\right) \subseteq \mathbb{B}_\delta (\hat{x}_{k-l}),
    \end{equation}
    where $ x_{k-l}$ is the solution of \eqref{eq:sampled-data_system disturbance} at time step $k-l$.
\end{definition}
The notion of $\delta$-bounded differential observability requires that we can bound the system state, given the uncorrupted input-output data.
This property has already been established for many nonlinear systems using various approaches, including deterministic extended Kalman filters \cite{boutayeb1997convergence}, observer design \cite{arcak2003observer}, and the recent Savitzky-Golay filtering-based observer design \cite{bunton2024confidently,silvestre2024nonlinear}. 
 In the presence of attacks, the notion of plausible states follows.
\begin{definition}[Plausible initial states] \label{def:plausible state disturbance}
     Given input-output data $(U_{k-l:k-1}, Y_{k-l:k})$ of the system \eqref{eq:sampled-data_system disturbance}, we call  $z_{k-l}$ \emph{a plausible initial state} if a sequence $\{e_j\}_{j=k-l}^k$ satisfying Assumption \ref{ass:s attacks} and a sequence $\{w_j\}_{j=k-l}^{k-1}$ satisfying condition \eqref{eq:w_k bound} exist such that the following equations hold
     \begin{equation} \label{eq:plausible condition disturbance}
    \begin{aligned}
        & z_{j+1} = F(z_j, u_j) + w_j,  ~j= k-l, \ldots, k-1,\\
        & y_{j} = c(z_{j}) + e_j,   ~j= k-l, \ldots, k.
    \end{aligned}
\end{equation}  
  We denote the set of all plausible initial states by $\myset{X}_k^{k-l}$.
 \end{definition}

Similar to the previous section, we introduce a system $\Sigma_{\Gamma}^D$ with partial measurements from a set of sensors $\Gamma \subseteq [p]$ as:
\begin{equation} \label{eq:disturbance gamma-system}
\Sigma_{\Gamma}^D: \left\{
\begin{aligned}
    x_{k+1} & = F(x_k, u_k) + w_k,  \\
     y_{k+1}^{\Gamma}  & = c_{\Gamma}(x_{k+1}) + e_{k+1}^\Gamma.
\end{aligned}    \right.
\end{equation}
For these systems, we use the following consistency condition.

\begin{definition}[Consistency condition]
      Suppose that system \eqref{eq:disturbance gamma-system} is $\delta$-bounded differentially observable. The input-output data $( U_{k-l:k-1}, Y_{k-l:k}^\Gamma)$ of the system \eqref{eq:disturbance gamma-system} is \emph{consistent} if variables $ w_{k-l}, \ldots, w_{k-1}, z$ exist such that the following holds.
    \begin{equation} \label{eq:consistency disturbance}
        \begin{aligned}
        & \mathcal{L}_{\Gamma}^D(U_{k-l:k-1}, Y_{k-l:k}^\Gamma) \subseteq \mathbb{B}_{\delta}(\hat{x}_{k-l}^\Gamma) \\
         &   z_{k-l} = z, ~\| z - \hat{x}_{k-l}^\Gamma\| \leq \delta, ~\| w_k \| \leq \bar{w}, \\
        & z_{j+1} = F(z_j, u_j) + w_j,   \textup{for } j= k-l, \ldots, k-1,\\
        & y_{j}^{\Gamma} = c_{\Gamma}(z_{j}) ,  \textup{for } j= k-l, \ldots, k. 
        \end{aligned}
    \end{equation}
\end{definition}

This consistency condition can be numerically checked, for example, by formulating a feasibility optimization program with decision variables $ w_{k-l}, \ldots, w_{k-1}, z$. Another practical approach is to first forward propagate the system state from $\hat{x}_{k-l}^\Gamma$, then compute the error between $y^\Gamma$ and the would-be measurements from the propagated states, and finally compare it to a small empirically obtained threshold. 

With the consistency condition, we develop similar results to last section for systems with process disturbances.
\begin{definition}[$r$-Sparse $\delta$-bounded observability] \label{def:sparse_observability disturbance}
    The sampled-data system \eqref{eq:sampled-data_system disturbance} is \emph{$r$-sparse  $\delta$-bounded (differentially) observable}, if  for any $\Gamma \in \mathbb{C}_{p}^{p-r}$, the system $\Sigma_{\Gamma}^D$ is $\delta$-bounded differentially observable of order $l$. 
\end{definition}

\begin{lemma}\label{lem:consistency inexact case}
     Let $\Gamma_1,\Gamma_2$ be two index sets and $\Gamma_1 \subseteq \Gamma_2 \subseteq [p]$. We have the following results.
     \begin{enumerate}
         \item If the system $\Sigma_{\Gamma_1}^D$ is $\delta$-bounded differentially observable, then the system $\Sigma_{\Gamma_2}^D$ is also $\delta$-bounded  differentially observable.
         \item Suppose that the systems $\Sigma_{\Gamma_1}^D$ and $\Sigma_{\Gamma_2}^D$ are both $\delta$-bounded differentially observable. Let the respective observability maps be $ \mathcal{L}_{\Gamma_1}^D, \mathcal{L}_{\Gamma_2}^D$.  If the input-output data $(U_{k-l:k-1}, Y_{k-l:k}^{\Gamma_2} )$ is consistent, then the input-output data $(U_{k-l:k-1}, Y_{k-l:k}^{\Gamma_1} )$ is also consistent, and the following relation holds:
     \begin{equation}
       \hspace{-5mm} \mathcal{L}_{\Gamma_1}^D(U_{k-l:k-1}, Y_{k-l:k}^{\Gamma_1} ) \cap \mathcal{L}_{\Gamma_2}^D(U_{k-l:k-1}, Y_{k-l:k}^{\Gamma_2} ) \neq \emptyset.
     \end{equation}
     \end{enumerate}
\end{lemma}
The proof of Lemma~\ref{lem:consistency inexact case} follows similar steps to those of Lemma~\ref{lem:consistency over two Gamma} and is neglected due to space limitations. The difference is that now we show by contradiction the set $\mathcal{L}_{\Gamma_2}^D(U_{k-l:k-1}, Y_{k-l:k}^{\Gamma_2} )$ must contain at least one point that fulfills the consistency condition for the data $(U_{k-l:k-1}, Y_{k-l:k}^{\Gamma_1})$.

We are ready to present results that gives an over-approximation of the plausible states. 
\begin{proposition} \label{prop:s_sparse disturbance}
    If system \eqref{eq:sampled-data_system disturbance} is $s$-sparse $\delta$-bounded observable, then the set of plausible initial states satisfies
    \begin{equation} \label{eq:X_t_t-l disturbance}
    \begin{aligned}
        \myset{X}_{k}^{k-l} \subseteq \bigcup_{\Gamma \in \Lambda}& \mathcal{L}_{\Gamma}^D ( U_{k-l:k-1}, Y_{k-l:k}^\Gamma ),
    \end{aligned}
    \end{equation}
    where $\Lambda := \{ \Gamma\in \mathbb{C}_{p}^{p-s} \text{ and }  (U_{k-l:k-1}, Y_{k-l:k}^\Gamma) \text{ is consistent} \}$.
\end{proposition}

\begin{proof}
    Based on Definition~\ref{def:plausible state disturbance}, any $z_{k-l} \in \myset{X}_{k}^{k-l}  $ has a corresponding $\{e_j\}_{j=k-l}^k$ satisfying Assumption~\ref{ass:s attacks} and $\{w_j\}_{j=k-l}^{k-1}$ satisfying condition \eqref{eq:w_k bound}. Let $\Gamma_A$ be the attacked sensors, then we know $|\Gamma_A|\leq s$ from Assumption~\ref{ass:s attacks}, so we can find an index set $\Gamma$ of size $|\Gamma| = p-s$ such that $\Gamma\subseteq [p]\setminus\Gamma_A$. 
    Furthermore, we construct a virtual system with partial measurement $\Sigma_{\Gamma}^{D,\prime}$ as in \eqref{eq:disturbance gamma-system}, which starts at $z_{k-l}$ at time $t_{k-l}$, and is subject to the process noise $\{w_j\}_{j=k-l}^{k-1}$ and zero attack signals $\{e_j^\Gamma = 0, \forall j \} $. As this system is $\delta$-bounded differentiable observable, we establish that $z_{k-l} \in \mathcal{L}_{\Gamma}^D ( U_{k-l:k-1}, Y_{k-l:k}^\Gamma ) \subseteq \mathbb{B}_{\delta}(\hat{x}_{k-l}^\Gamma) $ for some $\hat{x}_{k-l}^\Gamma$.  Thus, the data $ (U_{k-l:k-1}, Y_{k-l:k}^\Gamma)$  is consistent by definition since $\{w_j\}_{j=k-l}^{k-1}$ and $z_{k-l}$ satisfy equation \eqref{eq:consistency disturbance}. As a result, $z_{k-l}$ belongs to the set on the right-hand side, and we conclude that the set inclusion holds.
\end{proof}

A special result for the $2s$-sparse $\delta$-bounded observability case is established below.
\begin{corollary}
    If system \eqref{eq:sampled-data_system disturbance} is $2s$-sparse $\delta$-bounded observable, then
    the set of plausible initial states satisfies
    \begin{equation} \label{eq:2s-plausible_set disturbance}
        \myset{X}_{k}^{k-l} \subseteq \bigcup_{\Gamma \in \Lambda} \mathcal{L}_{\Gamma}^D ( U_{k-l:k-1}, Y_{k-l:k}^\Gamma ) \subseteq \mathbb{B}_{4\delta}(x_{k-l})
    \end{equation}
    where $\Lambda := \{ \Gamma\in \mathbb{C}_{p}^{p-s} \text{ and }  (U_{k-l:k-1}, Y_{k-l:k}^\Gamma) \text{ is consistent} \}$.
\end{corollary}
\begin{proof}
    The first set inclusion is true since $2s$-sparse $\delta$-bounded observability is a special case of $s$-sparse $\delta$-bounded observability. To prove the second set inclusion, consider any $z_{k-l}$ from the union. Let $\Gamma_2\in \Lambda$ be its set of sensors such that $z_{k-l}\in\mathcal{L}_{\Gamma_2}^D(U_{k-l:k-1}, Y_{k-l:k}^{\Gamma_2} )$ and the data is consistent.
    There exists a sub-index set $\Gamma_1\subseteq \Gamma_2 $ with $|\Gamma_1|= p -2s$ containing only attack-free sensors. 
    Then from $2s$-sparse $\delta$-bounded observability and Lemma~\ref{lem:consistency inexact case}, the observability map $  \mathcal{L}_{\Gamma_1}^D$ exists, and
    $\mathcal{L}_{\Gamma_1}^D(U_{k-l:k-1}, Y_{k-l:k}^{\Gamma_1} ) \cap \mathcal{L}_{\Gamma_2}^D(U_{k-l:k-1}, Y_{k-l:k}^{\Gamma_2} ) \neq \emptyset$.
    Note that $\mathcal{L}_{\Gamma_1}^D(U_{k-l:k-1}, Y_{k-l:k}^{\Gamma_1} )$ must contain the actual state $x_{k-l}$. As both sets can be over-approximated by a $\delta$-radius ball, any state estimate $z_{k-l}\in \mathcal{L}_{\Gamma_2}^D(U_{k-l:k-1}, Y_{k-l:k}^{\Gamma_2})$ must be within a $4\delta$ distance from $x_{k-l}$, concluding the proof.
\end{proof}

We further show that the current plausible states are bounded under mild conditions, as stated below.
\begin{proposition} \label{prop:bounded plausible states}
    Suppose that system \eqref{eq:sampled-data_system disturbance} is $s$-sparse $\delta$-bounded observable and the state transition map $F(x,u)$ is Lipschitz continuous in $x$ uniformly in $u$. Let $L$ be the Lipschitz constant. The set of current plausible states $\myset{X}_k^k$ is bounded within a union of ball regions:
\begin{equation} \label{eq:bounded plausible states disturbance}
    \hspace{-0.32cm} \myset{X}_{k}^{k} \subseteq \bigcup_{\Gamma\in \Lambda} \mathbb{B}_{\delta'}(\hat{x}_{k}^\Gamma),
\end{equation}
where  $\delta' := \underbrace{g\circ \ldots g\circ g}_{l \text{ times}}(\delta)$ with the function $g(s) := Ls + \bar{w}$, and $\hat{x}_{k}^\Gamma :=F \ldots F(F(\hat{x}_{k-l}^\Gamma, u_{k-l}),u_{k-l+1})\ldots u_{k-1}) $, $\Lambda := \{ \Gamma\in \mathbb{C}_{p}^{p-s} \text{ and }  (U_{k-l:k-1}, Y_{k-l:k}^\Gamma) \text{ is consistent} \}$, and $\hat{x}_{k-l}^\Gamma$ comes from the observability map $\mathcal{L}_\Gamma^D$.
\end{proposition}

\begin{proof}
    We first consider the set of plausible states at the discrete-time $k-l+1$. Along system dynamics, we know
    \begin{equation}
    \hspace{-3mm} \begin{aligned}
       \myset{X}_{k}^{k-l+1} \subseteq \{ & z_{k-l+1}\in \mathbb{R}^n:  \exists z_{k-l}\in \myset{X}_{k}^{k-l}, \|w_{k-l}\| \leq \bar{w} \\
       &\textup{ s.t. } z_{k-l+1}= F(z_{k-l},u_{k-l}) + w_{k-l}  \}.
    \end{aligned}
    \end{equation}
     One derives, for any $z_{k-l+1}\in \myset{X}_{k}^{k-l+1} $,
    \begin{equation}
    \hspace{-3mm} \begin{aligned}
        \| z_{k-l+1}  &- F(\hat{x}_{k-l}^\Gamma,  u_{k-l}) \| \\
        &  = \| F(z_{k-l},u_{k-l}) + w_{k-l} -  F(\hat{x}_{k-l}^\Gamma, u_{k-l})\| \\
        & \leq \| F(z_{k-l},u_{k-l})  - F(\hat{x}_{k-l}^\Gamma, u_{k-l})\| + \| w_{k-l} \| \\
        & \leq L \| z_{k-l} - \hat{x}_{k-l}^\Gamma\| + \| w_{k-l} \| \leq  L\delta + \bar{w}
    \end{aligned}
    \end{equation}
    for a certain $\Gamma\in \mathbb{C}_{p}^{p-s} $ with consistent $(U_{k-l:k-1}, Y_{k-l:k}^\Gamma) $.
    Thus, the result in \eqref{eq:bounded plausible states disturbance} follows from recursively repeating above analysis for $l$ times.
\end{proof}

Leveraging a robust CBF formulation from \cite{jankovic2018robust}, a secure safety filter for time  $t\in [t_k, t_{k+1})$ is given by
\begin{equation} \label{eq:ssf disturb}
    \begin{aligned}
        &  u_{\textup{safe}}(t) =    \argmin_{u} \ \|  u - u_{\textup{nom}}(t_k)\|^2 \\
        \mathrm{s.t. } \ &   h(F(x,u)) - h(x) \geq - \gamma(h(x)) + \epsilon + \epsilon_1 \\
        & \hspace{3.5cm} \textup{ for } x\in \{ \hat{x}_{k}^\Gamma \}_{\Gamma\in \Lambda} .
    \end{aligned}
\end{equation}
Here  $\epsilon_1$ is chosen such that 
$\epsilon_1\geq L_1(L\delta' + \bar{w})$, $L_1$ is the Lipschitz constant of the function $h(\cdot)$,  $\Lambda $, $\hat{x}_{k}^\Gamma$, $L$, and $\delta'$ are defined in Proposition~\ref{prop:bounded plausible states}.

\begin{thm} \label{thm:safe_control disturbance}
   Suppose that Assumption~\ref{ass:s attacks} holds and $h$ is a zero-order CBF for  system~\eqref{eq:continuous-time systems}. When system \eqref{eq:plausible condition disturbance} is $s$-sparse $\delta$-bounded observable, system~\eqref{eq:continuous-time systems} is safe on set $\myset{C}$ with the secure safety filter \eqref{eq:ssf disturb} if it is always feasible. 
\end{thm}
\begin{proof}
    From Proposition~\ref{prop:bounded plausible states}, we know the state $x(t_k) \in \myset{X}_{k}^{k} \subseteq \bigcup_{\Gamma\in \Lambda} \mathbb{B}_{\delta'}(\hat{x}_{k}^\Gamma) $. Thus one derives 
    \begin{equation} \label{eq:h_change disturbance}
        \begin{aligned}
        h(& \phi(T;x(t_k),u)) = h(F(x(t_k),u) + w_k)\\
        &  \geq h(F(x(t_k),u))  -L_1 \bar{w} \\
        &  \geq h(F(\hat{x}_k^\Gamma,u))  -L_1 \bar{w}  - L_1 L \delta' \text{ for some } \Gamma\in \Lambda\\
        & \geq  (1-\gamma)h(\hat{x}_k^\Gamma) + \epsilon  
        \end{aligned}
    \end{equation}
    The first inequality holds because of condition \eqref{eq:w_k bound} and the Lipschitz continuity property of $h$. The second inequality follows from Proposition~\ref{prop:bounded plausible states}. The third inequality is obtained from enforcing the constraint in the safety filter \eqref{eq:ssf disturb}, which is feasible by assumption. From \eqref{eq:ssf disturb}, $h(F(\hat{x}_k^\Gamma,u))  \geq (1-\gamma)h(\hat{x}_k^\Gamma) + \epsilon  $ for each $k$. By recursive reasoning, we know $h(\hat{x}_k^\Gamma)$ remains positive if it starts positive. Thus $h( \phi(T;x(t_k),u))\geq \epsilon$, which implies that $h(\phi(t;x(t_k),u))\geq 0$ for $t\in [t_k,t_{k+1})$. Recursively applying above analysis, we conclude that the closed-loop continuous-time system is safe on set $\myset{C}$. 
\end{proof}

\section{Simulation Results}
We demonstrate our proposed secure safety filter on a unicycle model:
\begin{equation}
    \begin{aligned}
        \dot{p}_1  = v \cos(\theta), 
        \dot{p}_2 = v\sin(\theta),
        \dot{\theta} = \mu,
    \end{aligned}
\end{equation}
where  $p_1\in \mathbb{R},p_2\in \mathbb{R}$ are the $x$-, $y$-coordinate  positions, $\theta \in (-\pi,\pi]$ the heading angle, and $ (v, \mu)\in [-5,5]\times [-2,2]$ the linear and the angular velocities, respectively. Suppose that there are $5$  ``onboard sensors" measuring the $x$- and $y$- coordinate positions,  the heading angle $\theta$, and the relative distance as well as the bearing angle from the origin. The measurement model is given by
\begin{equation*}
\begin{aligned}
y = c(p_1,p_2,\theta) = \begin{pmatrix}
    p_1, p_2, \sqrt{p_1^2 + p_2^2},  \textup{atan2}(p_2, p_1), \theta
\end{pmatrix}
\end{aligned}
\end{equation*}
where $\textup{atan2}(p_2, p_1)\in (-\pi,\pi]$ returns the bearing angle from $(0,0)$ to $(p_1, p_2)$. One verifies that, if the unicycle does not stay still or always move horizontally or vertically, hereafter referred to as the observability singularity cases, the state of the continuous-time system can be uniquely determined by any $3$ out of the total $5$ sensor measurements and their first-order time derivatives. Thus, by tuning the sampling time $T$ and the data length $(l+1)$, the sampled-data system with inexact state transition map is $2$-sparse $\delta$-bounded observable when the system is not in the observability singularity cases. In the simulation, we take $T = 0.01$s, $l = 25$ steps.

\begin{figure}[th]
    \centering
    \includegraphics[width=\linewidth]{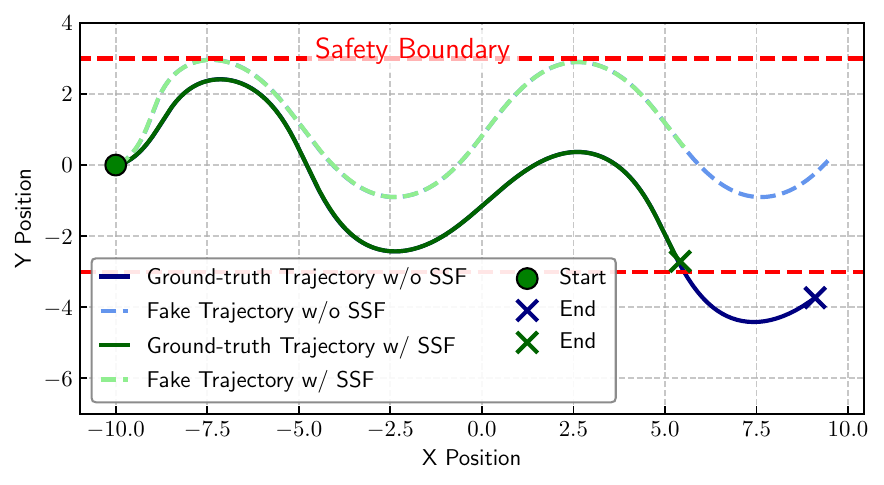}
    \caption{Unicycle trajectories with and without secure safety filter}
    \label{fig:traj}
\end{figure}

\begin{figure}[th]
    \centering
    \includegraphics[width=\linewidth]{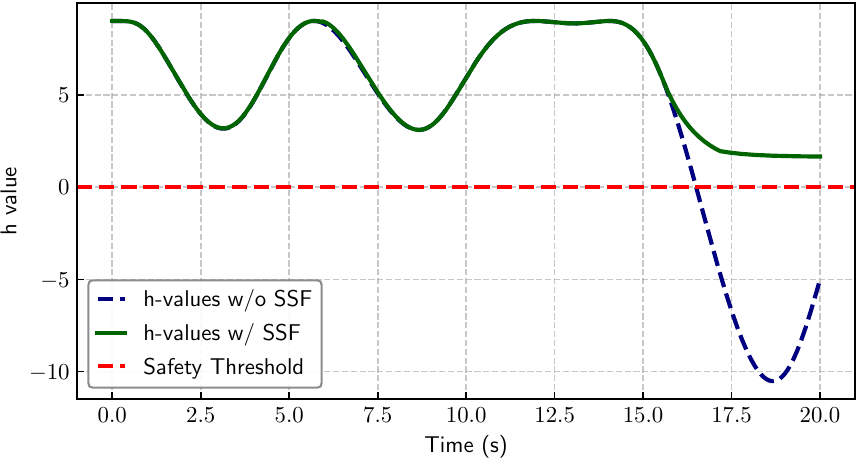}
    \caption{History of CBF values with and without secure safety filter}
    \label{fig:h_history}
\end{figure}

\begin{figure}[th]
    \centering
    \includegraphics[width=\linewidth]{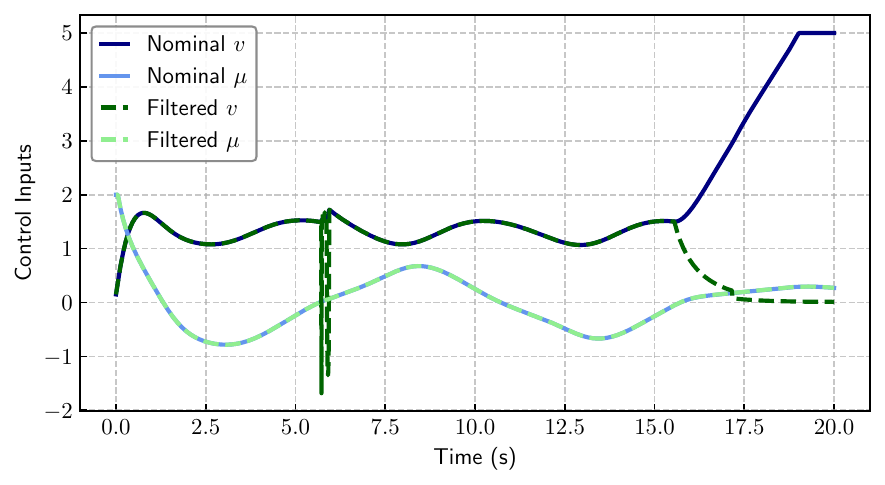}
    \caption{History of nominal input and safe input}
    \label{fig:u_history}
\end{figure}

In this simulation\footnote{An implementation code is available at \url{https://github.com/xiaotan-git/ssf_nonlinear_systems}.},  sensors $1$ and $2$ are subject to spoofing attacks, which is unknown to system designers. The attacked signals are generated as follows. Let $x_0 = (p_1(0), p_2(0),\theta(0) ) = (-10,0.0,-0.1)$ be the actual initial state at time $t = 0$s. The attacker tries to convince the controller a fake initial state $x_{\text{fake},0} = (-10,0.0,0.1)$. Based on input signals and system dynamics, the attacker simulates a fake system trajectory $x_{\text{fake}}(t)$ starting from $x_{\text{fake},0}$. For an attacked sensor $i$, the attacked signal $e^i(t)$ is designed such that the corresponding measurement $y^i(t) = c_i(x_{\text{fake}}(t))$.

We consider a scenario where the nominal control signal is computed remotely and only has access to $y^1(t)$ and $y^2(t)$. The nominal control signal aims to drive the unicycle to follow a curved path $(a_0 t + a_1, a_2 \sin(a_3 t + a_4) + a_5) $ with $(a_0, a_1,a_2,a_3,a_4,a_5) = (1.0,-10,1.8,\pi/5,0.0,1.0)$. See the curve  Fake Trajectory w/o SSF in Figure \ref{fig:traj} for an illustration.  The onboard safety filter has access to all $5$ measurements and corrects the nominal control on-the-fly. The safety constraint is to keep the unicycle within a horizontal band $\{ x=(p_1,p_2,\theta): h(x) = 3^2 - p_2^2 \geq 0\}$.

 In order to compute the observability map $\mathcal{L}_{\Gamma}^D$ for all sensor combinations $\Gamma\in \mathbb{C}_{5}^{3}$, we first derive analytical mappings from any $3$ sensor measurements and their first-order derivatives to the state, and then apply the Gaussian estimator from \cite{silvestre2024nonlinear} to obtain a numerical approximation of first-order time derivatives. After obtaining an initial state estimate, we then conduct consistency check by propagating the initial state estimate along system dynamics, compute the largest measurement error, and compare it to a threshold. 
 
 The closed-loop simulation results are reported in Figures \ref{fig:traj}, \ref{fig:h_history} and \ref{fig:u_history}. As shown in Figure \ref{fig:traj}, without applying our secure safety filter, the remote controller believes that the system (corresponding to the fake trajectory) behaves as expected, yet the unicycle in reality violates safety constraint. When our proposed secure safety filter is in place, we see that the unicycle movement is automatically corrected to be confined within the safety band. The history of zero-order CBF values in these two scenarios is shown in Figure~\ref{fig:h_history}. As seen from Figure~\ref{fig:u_history}, the secure safety filter takes effect only when the system is close to the safety boundary.

\section{Conclusion}

In this paper, we propose a secure safety filter design for general nonlinear systems under sensor spoofing attacks. Our approach extends secure safety filter design beyond linear systems by introducing exact and relaxed observability maps that abstract specific state estimation algorithms. We show how to generalize these observability maps to conduct secure state estimation. The secure safety filter is then designed by incorporating the secure state reconstructor with a control barrier function-based safety filter. Theoretical safety guarantees are provided for general nonlinear systems in the presence of sensor attacks. Finally, we validate the theoretical results numerically on a unicycle vehicle.

\bibliographystyle{IEEEtran}
\bibliography{IEEEabrv,references}

\end{document}